\documentclass{amsart} 

\usepackage[latin1,utf8]{inputenc}
\usepackage{amsfonts,bbm}
\usepackage{amssymb}
\usepackage{graphicx}
\usepackage{graphics}
\usepackage{lscape}
\usepackage{comment}
\usepackage{mathrsfs}
 
\newcommand{\QI}{q}
\newcommand{\E}{\mathbbm{E}}
\renewcommand{\P}{\mathbbm{P}}
\usepackage{amsmath}

\newtheorem{theorem}{Theorem}

\newtheorem{lemma}{Lemma}

\title[Processor Sharing Queue with batches]{{Asymptotic analysis of  the sojourn time of a batch in an $M^{[X]}/M/1$ Processor Sharing Queue}}

\author[F. Guillemin]{Fabrice Guillemin}
\address[F. Guillemin]{Orange Labs Networks Lannion, 2 avenue Pierre Marzin, 22307 Lannion Cedex, France}
\author[A. Simonian]{Alain Simonian}
\author[R. Nasri]{Ridha Nasri}
\address[A. Simonian and R. Nasri]{ Orange Labs, DATA-IA, Orange Gardens, 44 avenue de la République, CS 50010, 92326 Châtillon Cedex,  France}
\author[V. Quintuna]{Veronica Quintuna Rodriguez}
\address[V. Quintuna]{Orange Labs Networks Lannion, 2 avenue Pierre Marzin, 22307 Lannion Cedex, France}
\email{\{first\_name.last\_name\}@orange.com}

\begin{document}

\date{Version of \today}

\begin{abstract}
In this paper, we exploit  results obtained in an earlier study for the Laplace transform of the sojourn time $\Omega$ of an entire batch in the $M^{[X]}/M/1$ Processor Sharing (PS) queue in order to derive the asymptotic behavior of the complementary probability distribution function of this random variable, namely the behavior of $\P(\Omega>x)$ when $x$ tends to infinity. We precisely show that up to a multiplying factor, the behavior of $\P(\Omega>x)$ for large $x$ is of the same order of magnitude as $\P(\omega>x)$, where $\omega$ is the sojourn time of an arbitrary job is the system. From a practical point of view, this means that if a system  has to be dimensioned to guarantee  processing time for jobs then the system can  also guarantee processing times for entire batches by introducing a marginal amount of processing capacity. 

\end{abstract}


\maketitle


\section{Introduction}
\label{Sec:Intro}
In this paper, we consider an $M^{[X]}/M/1$ queue, where geometrically distributed batches (or bulks) of jobs arrive according to a Poisson process with rate $\rho$ and individual jobs require exponential service times with unit mean;  the probability that the  size $B$ of a batch is equal to $b$ is  $\P(B=b)=(1-q) q^{b-1}$ for some $q \in (0,1)$. We shall further assume that jobs are served according to the Processor Sharing discipline so that the server capacity is equally shared among all jobs present in the system. The system under consideration is denoted, for short, by $M^{[X]}/M/1$-PS queue.

The  occupancy of $M^{[X]}/M/1$-PS queue is well studied in the technical literature  \cite{grossharris,Klein0}. The sojourn time of an arbitrary job is however  more difficult as it involves complex correlations between jobs present in the system. The mean waiting time of a job in the system was computed  by Kleinrock \emph{et al} in \cite{Klein71} for arbitrary batch sizes via the derivation of an integral equation for the sojourn time of a tagged job conditioned on its service time. In the case of geometric batch sizes, the complete distribution of the sojourn time of a job has been derived in \cite{GQS}.

In this paper, we consider the sojourn time $\Omega$ of an arbitrary batch in the system, i.e., the time elapsed between the arrival of the batch and the departure of the last job of the batch considered. This quantity appears as a relevant performance parameter when considering the execution of batch of jobs in a cloud system; see notably the design of CloudRAN  \cite{veronica3,veronica2,cloudran}. In this context, the execution time of an entire batch is critical when considering the capacity of the system of encoding an entire frame within a tight time frame.  Processor sharing queue See also \cite{ayesta, bansal} for other applications.

The quantity $\Omega$ is defined as follows: Given a tagged batch with size $B = b$, $b \geqslant 1$, the sojourn time $\Omega$ equals the maximum
\begin{equation}
\Omega = \max_{1 \leqslant k \leqslant b} \omega_k
\label{defWbar}
\end{equation}
of the sojourn times $\omega_k$, $1 \leqslant k \leqslant b$, of the jobs of the tagged batch. Throughout this paper, we assume that 
$\rho+q < 1$ so that the system is stable.

The analysis of the random variable $\Omega$ is very complex because of the correlations between the sojourn times of the various jobs of a batch. The Laplace transform of $\Omega$ has nevertheless be derived in \cite{guillemin2020sojourn}. On the basis of the expression of the Laplace transform, we study in this paper the tail of the complementary distribution of random variable $\Omega$, precisely the asymptotic behavior of $\P(\Omega >x)$ when $x$ tends to infinity. 

An approximation for this quantity has been developed in \cite{itc31} by considering the residual busy period after the arrival of a tagged batch  and by assuming that the jobs of the tagged batch leave the system  at random among those jobs present in the system after the arrival of the tagged batch. In this paper, we are more accurate as we exploit the closed form of the Laplace transform of $\Omega$.

As in in \cite{itc31}, the asymptotic behavior of the complementary distribution of the sojourn time  $\omega$ of an arbitrary job serves as a baseline for comparison. For dimensioning purposes \cite{cloudran}, it is important to estimate how different are the sojourn times of a single job and that of an entire batch. 

Let us recall  that as shown in \cite{GQS}, we have as $x$ tends to infinity
\begin{equation}
    \label{tailomega}
\P(\omega>x) \sim \frac{2(1-\rho-q)(1-q)}{(\sigma_q^+)^2}e^{\frac{\sqrt{1-q}+ \sqrt{\rho}}{\sqrt{1-q}-\sqrt{\rho}}} \mathcal{D}_q(x),
\end{equation}
where
\begin{equation}
    \label{defQ}
    \mathcal{D}_q(x) = \frac{c_q(\rho)}{x^{\frac{5}{6}}}\exp\left(\sigma_q^+ x-b_q(\rho)  x^{\frac{1}{3}}\right),
\end{equation}
with $\sigma_q^+=-(\sqrt{\rho}-\sqrt{1-q})^2$ and
\begin{eqnarray*}
c_q(\rho)&=& \frac{1}{2^{\frac{1}{3}}\sqrt{3}}\left(\frac{\pi}{\sqrt{\rho(1-q)} }\right)^{\frac{5}{6}} , \\
b_q(\rho) &=& 3\left(\frac{\pi}{2}  \right)^{\frac{2}{3}}(\rho(1-q))^{\frac{1}{6}}.
\end{eqnarray*}

The major result of this paper states that $\P(\omega>x)$ and $\P(\Omega>x)$ have the same asymptotic behavior up to a constant prefactor. In other words,
$$
\frac{\P(\omega>x)}{\P(\Omega>x)}= O(1).
$$
This means that for extreme sojourn times, batches do not stay in the system much longer than jobs. Hence, if we dimension a system on the basis of jobs then batches will not experience too long sojourn times.

The organization of this paper is as follows: In Section~\ref{preliminary} we recall basic results from \cite{guillemin2020sojourn} and notably the Laplace transform of $\Omega$. The various terms appearing in this Laplace are inverted in Section~\ref{lapinversion}, which allows us to derive the asymptotic behavior of the cumulative distribution of random variable $\Omega$. Some concluding remarks are presented in Section~\ref{conclusion}. Technical proofs are deferred to a series of appendices.

\section{Preliminary results}
\label{preliminary}

We recall in this section the closed form of  the Laplace transform of the sojourn time $\Omega$ of an entire batch in the $M^{[X]}/M/1$ PS queue; this result has been obtained in \cite{guillemin2020sojourn}. To introduce this closed form, let us first define the polynomial $P(s;u)$ in $u$ by
\begin{equation}
    \label{defPTheta}
P(s;u) = u^2-(1+\rho+q+s) u +q+qs +\rho.
\end{equation}
The roots of $P(s;u)$ are
$$
U_q^{\pm}(s) = \frac{s+1+\rho+q \pm \sqrt{\delta(s)}}{2},
$$
where the discriminant $\delta(s) = s^2+2(1+\rho-q)s+(1-\rho-q)^2$. We note that $\delta(s)\geq 0$ for $s\leq \sigma_q^-$ or $s\geq \sigma_q^+$ and $\delta(s)<0$ for $s\in (\sigma_q^-,\sigma_q^+)$, where
$$
\sigma_q^\pm = -(\sqrt{1-q}\mp \sqrt{\rho})^2.
$$
The roots $U_q^\pm(s)$ satisfy for $s>0$
$$
q<U_q^-(s)<1<U_q^+(s).
$$
In the following, we shall consider $U_q^\pm(s)$ for complex values of $s$. In that case, $\sqrt{\delta(s)}$ shall denote the analytic extension of the function $\sqrt{\delta(s)}$ in the cut plane $\mathbb{C}\setminus[\sigma_q^-,\sigma_q^+]$ such that $\sqrt{\delta(0)}>0$.

Let us further introduce the function $R(s;\xi)$ defined by
\begin{equation}
R(s;\xi) = \left(1-\frac{\xi}{U_q^-(s)}\right)^{C_q^-(s) -1}  \left(1-\frac{\xi}{U_q^+(s)}\right)^{C_q^+(s)-1} ,
    \label{defR}
\end{equation}
where
\begin{equation}
C_q^+(s)= - \frac{U_q^-(s) - \QI}{U_q^+(s)  - U_q^-(s)}<0,
\quad
C_q^-(s)= 1-C_q^+(s)=- \frac{U_q^+(s) - \QI}{U_q^-(s) -U_q^+(s)}>1.
\label{C+-}
\end{equation}
On the basis of the function $R(s;t)$,  we further define  the function
$$
\label{defK}
\mathfrak{R}(s;u,v) = \frac{R(s;v)}{R(s;u)}= \left(\frac{v-U_q^-(s)}{u-U_q^-(s)}\right)^{C_q^-(s)-1} \left(\frac{v-U_q^+(s)}{v-U_q^+(s)}\right)^{C_q^+(s)-1}
$$
to simplify the notation.

For $s>0$, let us now consider function  $\mathcal{T}(s;w)$, which is solution  to the equation
\begin{equation}
    \label{defT}
 1-T + w T^{1-C_q^+(s)} =0   
\end{equation}
such that $\mathcal{T}(s;0)=1$. This function is instrumental in the computation of the Laplace of the random variable $\Omega$. For the properties of this function, we refer to  \cite{Polya} . In particular, we have the series expansion
$$
\mathcal{T}(s;w)= 1 +\sum_{n=1}^\infty \binom{n(1-C_q^+(s))}{n-1}\frac{w^n}{n}.
$$

To compute the Laplace transform of the  random variable $\Omega$, we introduce the bivariate generating function
\begin{equation}
  E(s;u,v) = \sum_{n=0}^\infty\sum_{b=1}^\infty e^*_{n,b}(s) U_q^n v^b  ,
\end{equation}
where $e^*_{n,b}(s)$ is the Laplace transform of the random variable $\Omega_{n,b}$ denoting the sojourn time of a batch of size $b$ given that there are $n$ jobs in the system upon the arrival of the batch. In other words, 
$$e^*_{n,b}(s) \stackrel{def}{=} \mathbb{E}(e^{-s \Omega_{n,b}}). $$
The above generating function is defined for $(u,v)\in \mathbb{D}^2$, where $\mathbb{D}$ is the unit disk in $\mathbb{C}$, and $\Re(s) \geq 0$ since by definition $e^*_{n,b}(s) \leq 1$.

Associated with  function $E(s;u,v)$ is the function $F(s;u,v)$ defined by
\begin{equation}
F(s;u,v) = 
\left\{
\begin{array}{ll}
\displaystyle \frac{E(s;u,v) - E(s;\QI,v)}{u-\QI}, & \Re(s) \geqslant 0, \; u \in \mathbb{D} \setminus \{\QI\}, \; v \in \mathbb{C},
\\ \\
\displaystyle \displaystyle \frac{\partial E}{\partial u}(s;\QI,v), &
\Re(s) s \geqslant 0, \; u = \QI,  \; v \in \mathbb{C}.
\end{array} \right.
\label{defFTheta}
\end{equation}
By definition, function $F$ is clearly analytic in 
$\{s \; \vert \; s > 0\} \times \mathbb{D} \times \mathbb{D}$ and it is obviously equivalent to determine either function $E$ or $F$. As shown in \cite{guillemin2020sojourn}, function $F(s;u,v)$ is solution to the following linear first order partial differential equation
\begin{multline}
 u P(s;u) \, \frac{\partial F}{\partial u}(s;u,v) + 
v \left [ \rho(1-\QI) - (s+1+\rho-v)(u-\QI) \right ] \, 
\frac{\partial F}{\partial v}(s;u,v) \\
+ \; \left [ u(u-s-1-\rho) + (u-\QI)(u+v) \right ] \ F(s;u,v)  +  L(s;u,v) = 0 
\label{PDEF0}
\end{multline}
{with polynomial $P$ introduced in Equation~\eqref{defPTheta}, and}
\begin{equation}
L(s;u,v) =L_0(u,v) + (u+v) \, E(s;\QI,v) -v (s+1+\rho-v)  
\frac{\partial E}{\partial v}(s;\QI,v),
\label{defL}
\end{equation}
{where}
$$
L_0(u,v) = \frac{v}{(1-u)}.
$$

To ensure the analyticity of the solution to Equation~\eqref{PDEF0}, function $E(s;q,v)$ shall be defined as follows.

\begin{lemma}
The function $E(s;q,v)$ is given by
\begin{equation}
    \label{eqvexp}
E(s;q,v) =  \frac{Q_0(s;v)}{(U_q^+(s)-U_q^-(s))P(s;v)}  \int_0^{U_q^-(s)}\Psi_0(s;\mathcal{T}(s;x(s) R(s;\xi)X(s;v))) \frac{d \xi}{1-\xi},
\end{equation}
where $x(s) = 1-\frac{U_q^-(s)}{U_q^+(s)}$,  the first order polynomial 
\begin{equation}
    \label{defQ0}
    Q_0(s;v) = U_q^+(s) U_q^-(s) -q v,
\end{equation}
the function
\begin{equation}
    \label{defX}
    X(s;v) =  \frac{v}{v-U_q^-(s)}\left(\frac{1-\frac{v}{U_q^-(s)}}{1-\frac{v}{U_q^+(s)}}  \right)^{C_q^+(s)} = \frac{-U_q^+(s) v}{P(s;v)R(s;v)}
\end{equation}
and the function
\begin{equation}
        \label{defPsi0}
\Psi_0(s;t) = \frac{t(1-t)}{(C_q^+(s) t +1-C_q^+(s))^3}.
\end{equation}
\end{lemma}

Before proceeding further, it is worth noting that 
\begin{equation}
    \label{derivePsi0}
    \frac{\partial \Psi_0}{\partial v}(\mathcal{T}(s;x(s)R(s;\xi)X(s;v)) =-\frac{Q_0(s;v)}{vP(s;v)}\Psi_1(s;\mathcal{T}(x(s)R(s;\xi)X(v)),
\end{equation}

\begin{equation}
    \label{defPsi1}
 \Psi_1(x) = \frac{t(1-t)(1-2 t -C_q^+(1-t^2))}{(C_q^+ t +1-C_q^+)^5}.  
\end{equation}

Once the function $E(s;q,v)$ is known, it is possible  to compute the function $F(s;u,v)$ as follows.

\begin{lemma}
Function $F(s;u,v)$, which is the analytic solution $\mathbbm{D}\times \mathbbm{D}$ to Equation~\eqref{PDEF0}, is given by
\begin{align}
&F(s;u,v) =\frac{u}{(u-v)P(s;u)} \int_u^{U_q^-(s)} \left(1-Z\left(s;u,\frac{v}{u};y\right)\right) Z\left(s;u,\frac{v}{u};y\right)\frac{d y}{1-y} \label{Fuv} \\
& +\frac{u}{(u-v)P(s;u)}\int_{u}^{U_q^-(s)} \left(1-Z\left(s;u,\frac{v}{u};y\right)\right) L_1\left(s;y, y Z\left(s;u,\frac{v}{u};y\right)   \right)     \frac{d y}{y} \nonumber \\
& +  \frac{u}{(u-v)P(s;u)}\int_u^{U_q^-(s)}\left(1-Z\left(s;u,\frac{v}{u};y\right)\right) L_2\left(s;y Z\left(s;u,\frac{v}{u};y\right) \right)     \frac{d y}{y},\nonumber
\end{align}
 where
\begin{equation}
    \label{defZ}
    Z(s;u,v;y) =    \frac{v \mathfrak{R}(s;u,y)}{(1-v) + v \mathfrak{R}(s;u,y)},
\end{equation}
together with
\begin{multline}
    \label{defL1}
L_1(s;u,v)=  \frac{1}{U_q^+(s)-U_q^-(s)}  \left(u \frac{Q_0(s;v)}{P(s;v)}   + v \frac{Q_1(s;v)}{P(s;v)^2} \right) \\  \int_0^{U_q^-(s)}\Psi_0(s;\mathcal{T}(s;x R(s;\xi)X(s;v))) \frac{d \xi}{1-\xi} 
\end{multline}
and
\begin{multline}
    \label{defL2}
 L_2(s;v) = \\  \frac{(U_q^+(s)+U_q^-(s)-q-v)Q_0(s;v)^2}{(U_q^+(s)-U_q^-(s))P(s;v)^2} \int_0^{U_q^-(s)}  \Psi_1(s;\mathcal{T}(s;x R(s;\xi)X(s;v)))    \frac{d \xi}{1-\xi}
\end{multline}
with 
\begin{multline}
\label{defQ1}
Q_1(s;v) = (q^2-U_q^+(s)U_q^-(s))v^2 + 2U_q^+(s)U_q^-(s)(U_q^+(s)+U_q^-(s)-2q)v \\  -  U_q^+(s)U_q^-(s)((U_q^+(s)+U_q^-(s) -q)^2-U_q^+(s)U_q^-(s)),
\end{multline}
and $\Psi_0(s;t)$ and $\Psi_1(s;t)$ are  defined by Equation~\eqref{defPsi0} and \eqref{defPsi1}, respectively.
\end{lemma}

With the above notation, the  Laplace transform of the random variable $\Omega$ is given by the following result.

\begin{theorem}
The  Laplace transform of the random variable $\Omega$ is given by
\begin{equation}
\label{laptransfoW}
\mathbbm{E}(e^{-s \Omega}) = \frac{1-\rho-q}{q(\rho+q)}\left(\rho^2 F(s;\rho+q,q) + \frac{q^3+\rho(q s +\rho+2q(1-q))E(s;q,q)}{q +\rho+ qs -q^2}   \right),
\end{equation}
where the function $E(s;q,v)$ is given by Equation~\eqref{eqvexp} and the function $F(u,v)$ by Equation~\eqref{Fuv}.
\end{theorem}

By using the Laplace transform of $\Omega$, we determine the tail of the complimentary probability distribution function of $\Omega$ by inverting the various terms involved in Equation~\eqref{laptransfoW}.

\section{Laplace transform inversion}
\label{lapinversion}

\subsection{Main result}

We assume in this paper that Laplace transforms can be defined in $\mathbbm{C}\setminus [\sigma_q^-, \sigma_q^+]$. The objective of this section is to prove the following result. 

\begin{theorem}
\label{asympOmegath}
For large $x$, we have
\begin{equation}
    \label{asympOmegaeq}
    \P(\Omega>x) \sim -\frac{1-\rho-q}{\sigma^+_q (\rho+q)q}\sum_{j=1}^3 \eta_q^{(j)}(\rho)\mathcal{D}_q(x),
\end{equation}
where $\mathcal{D}_q(x)$ is defined by Equation~\eqref{defQ} and the quantities $\eta_q^{(j)}(\rho)$ for $j=1,2,3$ are defined by
\begin{eqnarray}
\eta_q^{(1)}(\rho)&=& \frac{2 q  ((\rho+q)(1-q)+2q\sqrt{\rho(1-q)})e^{\frac{\sqrt{1-q}}{\sqrt{1-q}-\sqrt{\rho}}} }{(1-q) (1-q-\sqrt{\rho(1-q)})}, \label{defeta1}\\
\eta_q^{(2)}(\rho) & = &  - \frac{2q (\rho+q) e^{\frac{\sqrt{1-q}+\sqrt{\rho}}{\sqrt{1-q}-\sqrt{\rho}}}}{\sigma_q^+},\label{defeta2}\\
\eta_q^{(3)}(\rho) &=& \frac{2q(\rho+q)(\sqrt{\rho(1-q)}-\rho)}{\sigma_q^+(q+\sqrt{\rho(1-q)})}   e^{\frac{\sqrt{\rho}}{(q+\sqrt{\rho(1-q)})(\sqrt{1-q}-\sqrt{\rho})}}.\label{defeta3}
\end{eqnarray}
\end{theorem}

To establish the above result, we shall inverse the various Laplace transforms  appearing in Equation~\eqref{laptransfoW}. As a matter of fact, the Laplace transform of $\Omega$ can be decomposed as 
\begin{equation}
    \label{laplacedecompo}
\mathbbm{E}(e^{-s \Omega}) = \frac{\displaystyle 1-\rho-q}{\displaystyle q(\rho+q)}\left( \mathcal{I}_1^*(s) +  \mathcal{I}_2^*(s) +  \mathcal{I}_3^*(s)+ \frac{q^3}{q +\rho+ qs -q^2} \right),
\end{equation}
where
\begin{equation}
   \mathcal{I}_1(s)  =  \frac{\displaystyle \rho(q s +\rho+2q(1-q))}{\displaystyle q +\rho+ qs -q^2}E(s;q,q) \label{defI1s},
   \end{equation}
  \begin{equation} 
    \mathcal{I}_2^*(s)  =   \frac{\rho^2 u}{(u-v)P(s;u)} \int_{u}^{U_q^-(s)} \left(1-Z\left(s;u,\frac{v}{u};y\right)\right) Z\left(s;u,\frac{v}{u};y\right)\frac{d y}{1-y},     \label{defI2s}
    \end{equation}
    and
    \begin{multline}
    \mathcal{I}_3^*(s)  = \frac{\rho^2 u}{(u-v)P(s;u)}\int_u^{U_q^-(s)} \left(1-Z\left(s;u,\frac{v}{u};y\right)\right) L_1\left(s;y, y Z\left(s;u,\frac{v}{u};y\right)   \right)     \frac{d y}{y} \\
 +  \frac{\rho ^2 u}{(u-v)P(s;u)}\int_u^{U_q^-(s)}\left(1-Z\left(s;u,\frac{v}{u};y\right)\right) L_2\left(s;y Z\left(s;u,\frac{v}{u};y\right) \right)     \frac{d y}{y} \label{defI3s}
\end{multline}
for $u=\rho+q$ and $v=q$.

We shall study the inverse functions of  the various Laplace transforms appearing in the right hand side of Equation~\eqref{laplacedecompo} and  determine their behavior for large values of their argument $x$. It is clear that from \cite{GQS,itc31}, the behavior is dominated by $e^{\sigma_q^+x}$ in the sense that
$$
\lim_{x \to \infty} \frac{\log(\P(\Omega>x))}{x}=\sigma_q^+.
$$. 

From Equation~\eqref{laplacedecompo}, we can observe  that the point $s_q^*$ defined by
$$
s_q^* = -\frac{\rho+q(1-q)}{q}
$$
is a potential singularity for the Laplace transform of $\Omega$. It is however easily checked that $s_q^*<\sigma_q^-$. The point $s_q^*$ is in the slit $(\sigma_q^-,\sigma_q^+)$ if and only if $2q > \sqrt{\rho(1-q)}$. The term
$$
\frac{q^3}{q +\rho+ q s -q^2} 
$$
gives rise by Laplace transform inversion to the term
$$
-\frac{q^2}{s_q^*} e^{s_q^* x}.
$$
This term is negligible when compared against $e^{\sigma_q^+ x}$ for large $x$.

For the inversion of the other terms appearing in Equation~\eqref{laplacedecompo}, we shall use the following two technical lemmas. In their proof, we define for $\zeta>0$
$$
\mathcal{K}_\zeta(v,\xi) = \left(\frac{\xi-v}{q-v}\right)^{\mathcal{C}(v)}
$$
with 
\begin{equation}
    \label{defCv}
    \mathcal{C}(v) = \frac{(v-q)^2}{\rho(1-q)-(v-q)^2}
\end{equation}
so that as shown in \cite{GQS}, we have
$$
\mathfrak{R}(s;\zeta,\xi) = \mathcal{K}_\zeta(U_q^+(s),\xi)\mathcal{K}_\zeta(U_q^-(s),\xi).
$$
The proof of the following lemma is given in Appendix~\ref{prooflammatech1}.

\begin{lemma}
\label{lemmatech11}
 For $\theta\in [0,\pi]$ and $t\in [0,1]$, we set
\begin{equation}
    \label{defstheta}
    s(\theta) = -1+q-\rho+2 \sqrt{\rho(1-q)}\cos{\theta}\in [\sigma_q^-,\sigma_q^+].
\end{equation}
and
\begin{multline}
    \label{defxithetat}
\xi(\theta,t) =  U_q^-(s(\theta)+i0)+t(U_q^+(s(\theta)+0i)-U_q^-(s(\theta)+0i)) \\ \in [ U_q^-(s(\theta)+i0), U_q^+(s(\theta)+i0)].
\end{multline}
If $\zeta< q+\sqrt{\rho(1-q)}$ then for all $t\in(0,1)$, the quantity  $\mathfrak{R}(s+0i; \zeta,\xi(\theta,t)) $  tends to 0 when $\theta$ tends to 0. In addition, we have
\begin{equation}
    \label{intR}
\int_{U_q^+(s(\theta)+0i)}^{U_q^-(s(\theta)+0i)}  \mathfrak{R}(s(\theta)+0i;\zeta,y)\frac{d y}{1-y} =- \frac{ m(\zeta,\theta) \pi i  e^{(\varepsilon(\zeta)( \Phi(\zeta,\theta) -\frac{\pi}{2})+\Phi(1,\theta) )  \cot\theta}   }{m(1,\theta) \cosh\left(\frac{\pi \cot\theta}{2}\right)},
\end{equation}
where 
\begin{equation}
    \label{defmtheta}
    m(\zeta,\theta) = \sqrt{(\zeta-q)^2+\rho(1-q) -2(\zeta-q)\sqrt{\rho(1-q)}\cos\theta},
\end{equation}
$\varepsilon(\zeta) =\mathrm{sgn} (  q+\sqrt{\rho(1-q)} -\zeta)$ with $\mathrm{sgn}(x)$ being equal to 1 if $x>0$ and $-1$ if $x<0$, and $\Phi(\zeta,\theta) \in \left[0,{\pi}\right]$ is such that
\begin{align}
\label{defPhitheta}
\cos \Phi(\zeta,\theta) &=\left\{ \begin{array}{ll}
\frac{\zeta -\sqrt{\rho(1-q)}\cos\theta-q}{m(\zeta,\theta)} & \mbox{if } \zeta> q+\sqrt{\rho(1-q)},\\ \\
\frac{\sqrt{\rho(1-q)}\cos\theta+q -\zeta}{m(\zeta,\theta)}   & \mbox{if } \zeta< q+\sqrt{\rho(1-q)},
\end{array}
\right.\quad \\  \sin \Phi(\zeta,\theta) &= \frac{{\sqrt{\rho(1-q)} \sin\theta}}{m(\zeta,\theta)}. \nonumber
\end{align}
\end{lemma}

The second technical lemma, that we shall use in the proofs, gives the asymptotic behavior of an expression appearing when analyzing the asymptotic behavior of the Laplace inverse functions.

\begin{lemma}
\label{tail}
For large $x$, 
\begin{equation}
    \label{approxtech2}
e^{\sigma_q^+x }\int_0^\pi \theta e^{-\frac{\pi}{\theta} -\sqrt{\rho(1-q)} x \theta^2}d\theta \sim \mathcal{D}_q(x),
\end{equation}
where $\mathcal{D}_q(x)$ is defined by Equation~\eqref{defQ}.
\end{lemma}

\begin{proof}
We can write
\begin{multline*}
\int_0^\pi \theta e^{-\frac{\pi}{\theta} -\sqrt{\rho(1-q)} x \theta^2}d\theta =  \\ \left ( \frac{\pi}{\sqrt{\rho(1-q)} x} \right)^{\frac{2}{3}} \int_0^{(\sqrt{\rho(1-q)} x)^{\frac{1}{3}}\pi^{\frac{2}{3}}}  t e^{-\pi^{\frac{2}{3}}(\sqrt{\rho(1-q)} x)^{\frac{1}{3}}  \left(\frac{1}{t}+t^2  \right)  }dt.
\end{multline*}
By setting $h(t)=\frac{1}{t} +t^2$  and using Laplace's method, we note that the stationary point for function $h(t)$ is $t^*=2^{-\frac{1}{3}}$ with $h(t^*)= 3. 2^{-\frac{2}{3}}$ and $h''(t^*)=6$ so that when $x$ tends to infinity
\begin{multline*}
\int_0^{(\sqrt{\rho(1-q)} x)^{\frac{1}{3}}\pi^{\frac{2}{3}}} t  e^{-\pi^{\frac{2}{3}}(\sqrt{\rho(1-q)} x)^{\frac{1}{3}}  \left(\frac{1}{t}+t^2  \right)  }dt \sim  \\
\frac{1}{2^{\frac{1}{3}}}e^{-\pi^{\frac{2}{3}}(\sqrt{\rho(1-q)} x)^{\frac{1}{3}} h(t^*)} \int_{-\infty}^\infty    e^{-\pi^{\frac{2}{3}}(\sqrt{\rho(1-q)} x)^{\frac{1}{3}} h''(t^*) \frac{y^2}{2} } dy = \\
\frac{1}{2^{\frac{1}{3}}}\sqrt{\frac{\pi^{\frac{1}{3}}}{3(\sqrt{\rho(1-q)} x)^{\frac{1}{3}}}} e^{-3\left(\frac{\pi}{2}\right)^{\frac{2}{3}}(\sqrt{\rho(1-q)} x)^{\frac{1}{3}} }.
\end{multline*}

Hence, for large $x$
$$
 \int_0^\pi \theta e^{-\frac{\pi}{\theta} -\sqrt{\rho(1-q)} x \theta^2}d\theta \sim   \frac{1}{2^{\frac{1}{3}}\sqrt{3}}\left(\frac{\pi}{\sqrt{\rho(1-q)} x}\right)^{\frac{5}{6}} e^{-3\left(\frac{\pi}{2}\right)^{\frac{2}{3}}(\sqrt{\rho(1-q)} x)^{\frac{1}{3}} } 
$$
and Equation~\eqref{approxtech2} follows.
\end{proof}

\subsection{Asymptotic behavior of inverse Laplace functions}

The  inverse functions of the Laplace transforms $\mathcal{I}_j^*(s)$, $j=1,\ldots,3$ are given by 
\begin{equation}
\label{invlaplace}
\mathcal{I}_j(x) = \frac{-1}{2i \pi} \int_{\sigma_q^-}^{\sigma_q^+} \Delta \mathcal{I}_j^*(s) e^{sx} ds,
\end{equation}
where $$\Delta \mathcal{I}^*(s) = \mathcal{I}^*(s+0i) - \mathcal{I}^*(s-0i).$$
To study the behavior of $\mathcal{I}_j(x)$ for large $x$, we are led to analyze the behavior of $\mathcal{I}^*_j(s)$ for $s$ in the neighbourhood of $\sigma_q^+$. By using the variable change~\eqref{defstheta}, we note that  $s(\theta)$ is close to $\sigma_q^+$ when the parameter $\theta$ is close to 0.

We first have the following result, the proof is deferred to Appendix~\ref{prooflemmaI1}

\begin{lemma}
\label{lemmaI1}
For large $x$, 
\begin{equation}
       \label{asympI1}
 \mathcal{I}_1(x) \sim \eta_q^{(1)}(\rho)  \mathcal{D}_q(x)   ,
\end{equation} 
where $ \eta_q^{(1)}(\rho)$ is defined by Equation~\eqref{defeta1}.
\end{lemma}

For the second term $\mathcal{I}_2^*(s)$, we have the following asymptotic estimate, which is proved in Appendix~\ref{prooflemmaI2}.

\begin{lemma}
\label{lemmaI2}
For large $x$, 
\begin{equation}
       \label{asympI2}
 \mathcal{I}_2(x) \sim \eta_q^{(2)}(\rho)  \mathcal{D}_q(x)   
\end{equation} 
where $ \eta_q^{(2)}(\rho) $ is defined by Equation~\eqref{defeta2}.
\end{lemma}

Finally, for the last term $\mathcal{I}_3^*(s)$, we have the following result proved in Appendix~\ref{prooflemmaI3}.

\begin{lemma}
\label{lemmaI3}
For large $x$, 
\begin{equation}
       \label{asympI3}
 \mathcal{I}_3(x) \sim  \eta_q^{(3)}(\rho)\mathcal{D}_q(x),
\end{equation} 
\end{lemma}
where  $\eta_q^{(3)}(\rho)$ is defined by Equation~\eqref{defeta3}.

\subsection{Proof of Theorem~\ref{asympOmegath}}

By definition, we have
$$
\E\left( e^{-s \Omega} \right) = -\int_0^\infty e^{-s x} d\P(\Omega>x).
$$
We hence deduce from Equations~\eqref{asympI1}, \eqref{asympI2} and \eqref{asympI3} that
$$
\frac{d}{dx} \P(\Omega>x) \sim   -\frac{1-\rho-q}{(\rho+q)q}\sum_{j=1}^3 \eta_q^{(j)}(\rho)\mathcal{D}_q(x)
$$
and Equation~\eqref{asympOmegaeq} easily follows via a simple integration.

The terms $\eta_q^{(1)}$ and $\eta_q^{(2)}$ are positive while the term $\eta_q^{(3)}$ is negative. We note that 
\begin{multline*}
\eta_q^{(1)}+ \eta_q^{(2)} = e^{\frac{\sqrt{1-q}}{\sqrt{1-q}-\sqrt{\rho}}}\left( \frac{4q\sqrt{\rho(1-q)}e^{\frac{\sqrt{1-q}}{\sqrt{1-q}-\sqrt{\rho}}} }{(1-q) (1-q-\sqrt{\rho(1-q)})} + \right.\\ \left. \frac{2q(\rho+q)}{1-q-\sqrt{\rho(1-q)}}\left( 1-  \frac{\sqrt{\rho(1-q)}}{q+\sqrt{\rho(1-q)}} e^{\left(\frac{\sqrt{\rho}}{q+\sqrt{\rho(1-q)}}-\sqrt{1-q}\right)\frac{1}{(\sqrt{1-q}-\sqrt{\rho})}} \right) \right).
\end{multline*}
The function $x \to \frac{x}{x\sqrt{1-q}+q}$ is increasing for $x>0$ and equal to $\sqrt{1-q}$ for $x=\sqrt{1-q}$. It follows that since $\sqrt{\rho}<\sqrt{1-q}$  the quantity
$$
\frac{\sqrt{\rho}}{q+\sqrt{\rho(1-q)}}-\sqrt{1-q}<0
$$
Hence,
$$
1-  \frac{\sqrt{\rho(1-q)}}{q+\sqrt{\rho(1-q)}} e^{\left(\frac{\sqrt{\rho}}{q+\sqrt{\rho(1-q)}}-\sqrt{1-q}\right)\frac{1}{(\sqrt{1-q}-\sqrt{\rho})}}>0
$$
and subsequently $\eta_q^{(1)}+ \eta_q^{(2)}>0$. It follows that the prefactor
$$
 -\frac{1-\rho-q}{\sigma^+_q (\rho+q)q}\sum_{j=1}^3 \eta_q^{(j)}(\rho)>  -\frac{1-\rho-q}{\sigma^+_q (\rho+q)q} \eta_q^{(1)}(\rho) = \frac{2(1-\rho-q)}{(\sigma_q^+)^2}e^{\frac{\sqrt{1-q}+ \sqrt{\rho}}{\sqrt{1-q}-\sqrt{\rho}}}.
$$

To conclude this section, it is worth noting that the tails of of the complementary probability distributions of the sojourn time $\omega$ of a job    v(given by Equation~\eqref{tailomega}) and that of an entire batch are of the same order of magnitude in the sense that
$$
\frac{\P(\omega>x)}{\P(\Omega>x)}  = 0(1) \leq 1-q.
$$

Finally, $q=0$, we recover the asymptotic behavior of the sojourn time $W$ of a job in the $M/M/1$-PS queue. As a matter of fact,
$$
\lim_{q\to 0} -\frac{1-\rho-q}{\sigma^+_q (\rho+q)q}\sum_{j=1}^3 \eta_q^{(j)}(\rho)=  \frac{2(1-\rho)}{(\sigma^+_0)^2} e^{\frac{1+\sqrt{\rho}}{1-\sqrt{\rho}}},
$$
which is the prefactor appearing in Equation~\eqref{tailomega} for $q=0$.

\section{Conclusion}
\label{conclusion}
By using the result established in \cite{guillemin2020sojourn} for the Laplace transform of the sojourn time $\Omega$ of a batch in the $M^{[X]}/M/1$ PS queue, we have derived the asymptotic behavior of $\P(\Omega>x)$ when $x$ tends to infinity. It turns out that up to a multiplying factor $\P(\Omega>x)$  and $\P(\omega>x)$, where $\omega$ is the sojourn time of a single job, exhibit the same behavior for lareg $x$. This indicates that for extreme values the difference between the sojourn time of a job and that of a batch are not significantly different. This could be taken into account when assessing the processing capacity of a server handling batches of jobs. 

\appendix

\section{Proof of Lemma~\ref{lemmatech11}}
\label{prooflammatech1}
Let us first assume that $\zeta > q+\sqrt{\rho(1-q)}$. For $\theta\in [0,\pi]$, we have
$$
U_q^+(s(\theta)+0i) = q+\sqrt{\rho(1-q)}e^{i\theta} \quad \mbox{and} \quad U_q^-(s(\theta)+0i) = q+\sqrt{\rho(1-q)}e^{-i\theta}.
$$
and for $t\in [0,1]$, we have
\begin{align*}
\mathcal{K}_{\zeta}(& U_q^+(s+ 0i),\xi(\theta,t)) = \left(\frac{\xi(\theta,t)-U_q^+(s+0i)}{\zeta-U_q^+(s+0i)}  \right)^{\mathcal{C}(U_q^+(s+0i))}\\
&= \left(\frac{-2i (1-t) \sqrt{\rho(1-q)} \sin\theta}{\zeta-q-\sqrt{\rho(1-q)}\cos\theta -i \sqrt{\rho(1-q)} \sin \theta}  \right)^{-\frac{1}{2}+\frac{i}{2}\cot\theta}
\\
&= m(\zeta,\theta)^{\frac{1}{2}-\frac{i}{2}\cot\theta} (2(1-t) \sqrt{\rho(1-q)}\sin\theta )^{-\frac{1}{2}+\frac{i}{2}\cot\theta} e^{-( \Phi(\zeta,\theta) -\frac{\pi}{2})(\frac{i}{2}+\frac{1}{2}\cot\theta)},
\end{align*}

Similarly,
\begin{multline*}
\mathcal{K}_{\zeta}(U_q^-(s+ 0i),\xi(\theta,t)=  \\ m(\zeta,\theta)^{\frac{1}{2}+\frac{i}{2}\cot\theta} (2t \sqrt{\rho(1-q)}\sin\theta )^{-\frac{1}{2}-\frac{i}{2}\cot\theta} e^{-( \Phi(\zeta,\theta) -\frac{\pi}{2})(-\frac{i}{2}+\frac{1}{2}\cot\theta)}
\end{multline*}
so that
\begin{multline*}
\mathfrak{R}(s(\theta)+0i; \zeta,\xi(\theta,t)) = \\ \frac{m(\zeta,\theta)}{2\sqrt{\rho(1-q)}\sin\theta} t^{-\frac{1}{2}-\frac{i}{2}\cot\theta}(1-t)^{-\frac{1}{2}+\frac{i}{2}\cot\theta}  e^{-( \Phi(\zeta,\theta) -\frac{\pi}{2})\cot\theta}.
\end{multline*}

In the case $\zeta < q+\sqrt{\rho(1-q)}$, similar computations show that
\begin{multline}
\label{repR}
\mathfrak{R}(s(\theta)+0i; \zeta,\xi(\theta,t)) = \\ \frac{m(\zeta,\theta)}{2\sqrt{\rho(1-q)}\sin\theta} t^{-\frac{1}{2}-\frac{i}{2}\cot\theta}(1-t)^{-\frac{1}{2}+\frac{i}{2}\cot\theta}  e^{( \Phi(\zeta,\theta) -\frac{\pi}{2})\cot\theta}.
\end{multline}
In this case, the module of $\mathfrak{R}(s+0i; \zeta,\xi(\theta,t)) $ is such that
$$
\left|\mathfrak{R}(s+0i; \zeta,\xi(\theta,t))  \right| \leq \frac{|\zeta-q|+\sqrt{\rho(1-q)}}{2\sqrt{\rho(1-q)}\sin\theta}t^{-\frac{1}{2}}(1-t)^{-\frac{1}{2}}  e^{( \Phi(\zeta,\theta) -\frac{\pi}{2})\cot\theta}.
$$
Since
$$
\Phi(\zeta,\theta)\sim\frac{\sqrt{\rho(1-q)}\theta}{\sqrt{\rho(1-q)}+q-\zeta}
$$
when $\theta$ tends to 0, we can conclude that if $\zeta< q+\sqrt{\rho(1-q)}$ then for all $t\in (0,1)$, $\mathfrak{R}(s+0i; \zeta,\xi(\theta,t)) $  tends to 0 when $\theta$ tends to 0.

By using the above results, we have by the variable change~\eqref{defxithetat}
\begin{multline*}
\int_{U_q^+(s(\theta)+0i)}^{U_q^-(s(\theta)+0i)}  \mathfrak{R}(s+0i;\zeta,y)\frac{d y}{1-y} = \\  
\int_0^1     \frac{ - m(\zeta,\theta)  e^{\varepsilon(\zeta)( \Phi(\zeta,\theta) -\frac{\pi}{2})\cot\theta} t^{-\frac{1}{2}-\frac{i}{2}\cot\theta}(1-t)^{-\frac{1}{2}+\frac{i}{2}\cot\theta}  i d t }{1- U_q^-(s(\theta)+0i) - t(U_q^+(s(\theta)+0i) - U_q^-(s(\theta)+0i))    }.
\end{multline*}

The integral
\begin{multline}
    \label{approxtech}
\int_{0}^{1}     \frac{t^{-\frac{1}{2}-\frac{i}{2}\cot\theta}(1-t)^{-\frac{1}{2}+\frac{i}{2}\cot\theta}  i dt }{1- U_q^-(s(\theta)+0i) -t(U_q^+(s(\theta)+0i)-U_q^-(s(\theta)+0i))} =\\ \frac{i}{1- U_q^-(s(\theta)+0i)} \Gamma\left(\beta(\theta)\right) \Gamma\left(1-\beta(\theta)\right) F\left(1, \beta(\theta), 1;z(\theta) \right),
\end{multline}
where
\begin{equation}
    \label{betatheta}
\beta(\theta) = \frac{1}{2}-\frac{i \cot\theta}{2}, \quad z(\theta) = \frac{U_q^+(s(\theta)+0i)-U_q^-(s(\theta)+0i)}{1- U_q^-(s(\theta)+0i)}
\end{equation}
and $F(\alpha,\beta,\gamma;z)$ denotes Gauss hypergeometric function \cite{Abramowitz}.

By using the fact that 
$$
\Gamma(1-\beta)\Gamma(\beta) = \frac{\pi}{\sin(\pi \beta)}
$$
and
$$
F(1,\beta,1;z) = \frac{1}{(1-z)^\beta},
$$
together with
$$
1-z(\theta)= e^{-2i \Phi(1,\theta)},
$$
where $\Phi(\zeta,\theta)$ is defined by Equation~\eqref{defPhitheta}, we obtain
\begin{multline}
\label{equationtech}
\int_{0}^{1}     \frac{t^{-\frac{1}{2}-\frac{i}{2}\cot\theta}(1-t)^{-\frac{1}{2}+\frac{i}{2}\cot\theta}  i dt }{1- U_q^-(s(\theta)+0i) -t(U_q^+(s(\theta)+0i)-U_q^-(s(\theta)+0i))} =\\ \frac{i}{1- U_q^-(s(\theta)+0i)} \frac{\pi}{\sin(\pi \beta(\theta))} e^{2i \Phi(1,\theta)\beta(\theta)} 
\end{multline}
and Equation~\eqref{intR} easily follows by using the fact that
$$
 \frac{1}{1- U_q^-(s(\theta)+0i)}= \frac{e^{-i \Phi(1,\theta)}}{m(1,\theta)}
$$
and
$$
\sin(\pi \beta(\theta))=\cosh\left(\frac{\pi \cot\theta}{2} \right).
$$
\section{Proof of Lemma~\ref{lemmaI1}}
\label{prooflemmaI1}

Let us first note that from Equation~\eqref{defI1s} 
$$
\mathcal{I}_1^*(s)  =   \frac{q s +\rho+2q(1-q)}{(U_q^+(s)-U_q^-(s))(1-q)}  \int_0^{U_q^-(s)}\Psi_0(s;\mathcal{T}(s;x(s) R(s;\xi)X(s;q))) \frac{d \xi}{1-\xi}.
$$
The Laplace inverse function $\mathcal{I}_1(x)$ is then defined by
\begin{equation}
\label{defmathcalE}
\mathcal{I}_1(x) = \frac{-1}{2i \pi} \int_{\sigma_q^-}^{\sigma_q^+} \Delta \mathcal{I}_1^*(s) e^{sx} ds.
\end{equation}
To study the behavior of $\mathcal{I}_1(x)$ for large $x$, we are led to consider the behavior of $\Delta \mathcal{I}_1^*(s)$ in the neighbourhood of $\sigma_q^+$.

When $s$ is close to $\sigma_q^+$,  $x(s\pm 0i)$ is close to 0 so that  for all $\xi \in [0, U_q^-(s\pm 0i)]$
\begin{equation}
\label{approxT}
\mathcal{T}(s\pm 0i;x(s\pm 0i) R(s\pm 0i;\xi)X(s\pm 0i;q))) \sim 1+ x(s\pm 0i) R(s\pm 0i;\xi)X(s\pm 0i;q)
\end{equation}
and then
\begin{multline}
    \label{approxPsi0}
\Psi_0(s\pm 0i;\mathcal{T}(s\pm 0i;x(s\pm 0i) R(s\pm 0i;\xi)X(s\pm 0i;q))) \sim  \\  - x(s\pm 0i) R(s\pm 0i;\xi)X(s\pm 0i;q).
\end{multline}

We have
$$
R(s\pm 0i;\xi)X(s\pm 0i;q)= \frac{-U_q^+(s\pm 0i) q}{\rho(1-q)} \mathfrak{R}(s\pm 0i;q;\xi),
$$
so  that for $s$ close to $\sigma_q^+$
$$
\mathcal{I}_1^*(s\pm 0i) \sim \frac{q(q s +\rho+2q(1-q))}{\rho(1-q)^2}\int_0^{U_q^-(s\pm0i)}\mathfrak{R}(s\pm 0i;q,\xi)\frac{d \xi}{1-\xi}.
$$
For $s\in[\sigma_q^-,\sigma_q^+]$, $U_q^+(s+0i)=\overline{U_q^-(s+0i)}=U_q^-(s-0i)$. It then follows that for $s$ close to $\sigma_q^+$ of the form~\eqref{defstheta} for small $\theta$
\begin{align*}
\Delta \mathcal{I}_1^*(s(\theta)) & \sim \frac{q(q \sigma_q^+ +\rho+2q(1-q))}{\rho(1-q)^2}\int_{U_q^+(s(\theta)+0i)}^{U_q^-(s(\theta)+0i)}   \mathfrak{R}(s(\theta)+ 0i;q,\xi)\frac{d \xi}{1-\xi}   \\
&\sim  \frac{q(q \sigma_q^+ +\rho+2q(1-q))}{\rho(1-q)^2} . \frac{- 2 i \pi   \sqrt{\rho(1-q)}   e^{\frac{\sqrt{1-q}}{\sqrt{1-q}-\sqrt{\rho}}}        }{\sqrt{1-q}(\sqrt{1-q}-\sqrt{\rho})}e^{ -\frac{\pi}{\theta}},
\end{align*}
where we have used Equation~\eqref{intR} for $\zeta=q$ and the fact that $\Phi(q,\theta)=\theta$, $m(q,\theta)=\sqrt{\rho(1-q)} $, and 
\begin{align*}
  m(1,\theta) &\sim \sqrt{1-q}(\sqrt{1-q}-\sqrt{\rho})  , \quad 
   \cosh\left(\frac{\pi \cot\theta}{2}\right)  \sim \frac{e^{\frac{\pi}{2\theta}}}{2},\\
   ( \theta-\Phi(1,\theta))\cot\theta & \sim \frac{\sqrt{1-q}}{\sqrt{1-q}-\sqrt{\rho}}.
\end{align*}
for small $\theta$.

From the above asymptotic estimate, we deduce that when $x$ is large, the inverse function $\mathcal{I}_1(x)$ defined by Equation~\eqref{defmathcalE} is such for large $x$  
\begin{multline*}
\mathcal{I}_1(x) \sim \\  \frac{q(q \sigma_q^+ +\rho+2q(1-q))}{\rho(1-q)^2} . \frac{2{\rho(1-q)}   e^{\frac{\sqrt{1-q}}{\sqrt{1-q}-\sqrt{\rho}}}        }{\sqrt{1-q}(\sqrt{1-q}-\sqrt{\rho})} e^{\sigma_q^+ x}   \int_0^\pi \theta e^{-\frac{\pi}{\theta} -\sqrt{\rho(1-q)} x \theta^2}d\theta.
\end{multline*}
and Equation~\eqref{asympI1} follows by using Lemma~\ref{tail}.

\section{Proof of Lemma~\ref{lemmaI2}}
\label{prooflemmaI2}

Given that $P(s;\rho+q)=-\rho s$, we have from Equation~\eqref{defI2s}
\begin{multline*}
\mathcal{I}_2^*(s) = \\-  \frac{(\rho+q)}{s } \int_{\rho+q}^{U_q^-(s)} \left(1-Z\left(s;\rho+q,\frac{q}{\rho+q};y\right)\right) Z\left(s;\rho+q,\frac{q}{\rho+q};y\right)\frac{d y}{1-y}.
\end{multline*}

By applying Lemma~\ref{lemmatech11} for $\zeta=\rho+q$, we have for $t\in [0,1]$ 
\begin{multline}
Z\left(s(\theta)\pm 0i; \rho+q,\frac{q}{\rho+q};\xi(\theta,t))\right) = \frac{q\mathfrak{R}(s(\theta)\pm 0i; \rho+q,\xi(\theta,t))}{\rho + q \mathfrak{R}(s(\theta)\pm 0i; \rho+q,\xi(\theta,t))}  \\ \sim \frac{q}{\rho} \mathfrak{R}(s(\theta)\pm 0i; \rho+q,\xi(\theta,t)).\label{approxZ}
\end{multline}
for small $\theta$. It follows that
\begin{align*}
\Delta \mathcal{I}_2^*(s(\theta))& \sim  -\frac{(\rho+q)}{s(\theta)} \int_{U_q^+(s(\theta)+0i)}^{U_q^-(s(\theta)+0i)}  Z\left(s(\theta)+0i;\rho+q,\frac{q}{\rho+q};y\right)\frac{d y}{1-y}\\
& \sim - \frac{(\rho+q) q}{\rho \sigma_q^+} \int_{U_q^+(s(\theta)+0i)}^{U_q^-(s(\theta)+0i)}    \mathfrak{R}(s(\theta)+ 0i; \rho+q,y) \frac{d y}{1-y}    
\end{align*}
for small $\theta$. By using Lemma~\ref{lemmatech11} for $\zeta = \rho+q$, we have for small $\theta$
\begin{align*}
\Delta \mathcal{I}_2^*(s(\theta))& \sim  
 \frac{(\rho+q) q}{\rho \sigma_q^+} . \frac{ m(\rho+q,\theta) \pi i  e^{(( \Phi(\rho+q,\theta) -\frac{\pi}{2})+\Phi(1,\theta) )  \cot\theta}   }{m(1,\theta) \cosh\left(\frac{\pi \cot\theta}{2}\right)} \\
 & \sim  
 \frac{(\rho+q) q}{\rho \sigma_q^+} . \frac{ 2 i \pi \sqrt{\rho}  e^{\frac{\sqrt{1-q}+\sqrt{\rho}}{\sqrt{1-q}-\sqrt{\rho}}} }{ \sqrt{1-q} } e^{-\frac{\pi}{2\theta}}  ,
\end{align*}
where we have used the fact that
\begin{align*}
m(\rho+q,\theta) &\sim \sqrt{\rho}(\sqrt{1-q}-\sqrt{\rho}),\\
m(\rho+q,\theta) &\sim \sqrt{1-q}(\sqrt{1-q}-\sqrt{\rho}),\\
(\Phi(\rho+q,\theta) +\Phi(1,\theta) )  \cot\theta &\sim \frac{\sqrt{1-q}+\sqrt{\rho}}{\sqrt{1-q}-\sqrt{\rho}}.
\end{align*}
for small $\theta$.

By Laplace inversion, we have
$$
\mathcal{I}_2(x) = -\frac{1}{2i\pi}\int_{\sigma_q^-}^{\sigma_q^+} \Delta \mathcal{I}_2^*(s)e^{s x} ds= -\frac{2 \sqrt{\rho(1-q)}}{2i\pi}\int_{0}^{\pi} \Delta \mathcal{I}_2^*(s(\theta))e^{s(\theta) x} \sin\theta d\theta 
$$
so that for large $x$
\begin{align*}
\mathcal{I}_2(x) \sim - \frac{2q (\rho+q) e^{\frac{\sqrt{1-q}+\sqrt{\rho}}{\sqrt{1-q}-\sqrt{\rho}}}}{\sigma_q^+} e^{x \sigma_q^+ } \int_0^\pi \theta e^{-x \theta^2 \sqrt{\rho(1-q)}-\frac{\pi}{2}}d\theta \\
\sim - \frac{2q (\rho+q) e^{\frac{\sqrt{1-q}+\sqrt{\rho}}{\sqrt{1-q}-\sqrt{\rho}}}}{\sigma_q^+} \mathcal{D}_q(x),
\end{align*}
where we have used Lemma~\ref{tail}.  This proves Equation~\eqref{asympI2}.

\section{Proof of Lemma\ref{lemmaI3}}
\label{prooflemmaI3}

The Laplace transform $\mathcal{I}^*_3(s)$ can be decomposed as 
$$
 \mathcal{I}^*_3(s) = \frac{u \rho^2}{(u-v)P(s;u)} I_3(s)
$$    
  with  $I_3(s)= I_{3,1}(s)+I_{3,2}(s)$, where
$$
I_{3,1}(s) =   \int_u^{U_q^-(s)} \left(1-Z\left(s;u,\frac{v}{u};y\right)\right) L_1\left(s;y, y Z\left(s;u,\frac{v}{u};y\right)   \right)     \frac{d y}{y}
$$
and
$$
I_{3,2}(s)= \int_u^{U_q^-(s)}\left(1-Z\left(s;u,\frac{v}{u};y\right)\right) L_2\left(s;y Z\left(s;u,\frac{v}{u};y\right) \right)     \frac{d y}{y}
$$
for $u=\rho+q$ and $v=q$. For these values of $u$ and $v$, we have
$$
 \frac{u \rho^2}{(u-v)P(s;u)} = -\frac{\rho+q}{s}.
$$

By Laplace inversion from Equation~\eqref{defI3s}, we are hence led to consider
$$
\mathcal{I}_3(x) =\frac{\rho+q}{2 i \pi} \int_{\sigma_q^-}^{\sigma_q^+} \Delta I_3(s) e^{sx} \frac{ds}{s},
$$
where $I_3(s) = I_{3,1}(s)+ I_{3,2}(s)$.

For $s\pm 0i$ close to $\sigma_q^+$, we use approximations~\eqref{approxT} and \eqref{approxPsi0} to obtain for fixed $u$ and $v$
\begin{equation}
    \label{approxL1}
    L_1(s\pm 0i;u,v) \sim \frac{v}{P(s,v)}\left(u \frac{Q_0(s;v)}{P(s;v)}+ v \frac{Q_1(s;v)}{P(s;v)^2}\right) \int_0^{U_q^-(s\pm0i)} \mathfrak{R}(s\pm 0i;v,\xi)\frac{d\xi}{1-\xi},
\end{equation}
where $L_1(s;u,v)$ is defined by Equation~\eqref{defL1}. Similarly, by using that for small $x$, $\Psi_1(1+x) \sim x$, we have
$$
L_2(s\pm 0i;v) \sim  - \frac{(1+\rho+s-v) v Q_0(s;v)^2}{P(s;v)^3} \int_0^{U_q^-(s\pm0i)} \mathfrak{R}(s\pm 0i;v,\xi)\frac{d\xi}{1-\xi},
$$
where $L_2(s,v)$ is defined by Equation~\eqref{defL2}.

For small $\theta$, we have by Approximations~\eqref{approxZ} and \eqref{approxL1}
\begin{multline*}
I_{3,1}(s(\theta)+0i)
\sim \\
 \int_{\rho+q}^{U_q^-(s(\theta)+0i)} \frac{ y Z\left(s(\theta)+0i;\rho+q,\frac{q}{\rho+q};y\right)  }{P(s(\theta)+0i,0)}\left(y \frac{Q_0(s(\theta)+0i;0)}{P(s(\theta)+0i;0)}\right) \frac{dy}{y}\\ \int_0^{U_q^-(s(\theta)+ 0i)} \mathfrak{R}\left(s(\theta)+ 0i; y Z\left(s(\theta)+0i;\rho+q,\frac{q}{\rho+q};y\right) ,\xi\right)\frac{d\xi}{1-\xi}.
\end{multline*}
The quantity $Z\left(s(\theta)+0i;\rho+q,\frac{q}{\rho+q};y\right)$ is exponentially decreasing in $1/\theta$ for all $y$ as $\theta$ goes to 0. Hence, 
$$
 \mathfrak{R}\left(s(\theta)+ 0i; y Z\left(s(\theta)+0i;\rho+q,\frac{q}{\rho+q};y\right) ,\xi\right) \sim  \mathfrak{R}\left(s(\theta)+ 0i; 0 ,\xi\right).
$$
It follows that for small $\theta$ by Approximation~\eqref{approxZ}
\begin{multline*}
I_{3,1}(s(\theta)+0i)
\sim 
\frac{q}{\rho(q\sigma_q^+ +q+\rho)}  \int_{\rho+q}^{U_q^-(s(\theta)+0i)} y  \mathfrak{R}(s(\theta)+ 0i; \rho+q,y) dy  \\ \int_{0}^{U_q^-(s(\theta)+0i)}  \mathfrak{R}\left(s(\theta)+ 0i; 0 ,\xi\right) \frac{d\xi}{1-\xi}.
\end{multline*}

By using similar arguments,  for small $\theta$
\begin{multline*}
  I_{3,2}(s(\theta)+0i) \sim 
-\frac{q(1+\rho+\sigma_q^+)}{\rho(q \sigma_q^+ +q+\rho)}  \int_{\rho+q}^{U_q^-(s(\theta)+0i)} \mathfrak{R}(s(\theta)+ 0i; \rho+q,y)  dy  \\ \int_{0}^{U_q^-(s(\theta)+0i)}  \mathfrak{R}\left(s(\theta)+ 0i; 0 ,\xi\right) \frac{d\xi}{1-\xi}  \end{multline*}
so that when $\theta$ is small
\begin{multline*}
 I_3(s(\theta)+0i) \sim \frac{q}{\rho} \int_{\rho+q}^{U_q^-(s(\theta)+0i)} \frac{y- \rho-1-\sigma_q^+}{q\sigma_q^+ +q+\rho} \mathfrak{R}(s(\theta)+ 0i; \rho+q,y)     dy  \\ \int_{0}^{U_q^-(s(\theta)+0i)}\mathfrak{R}\left(s(\theta)+ 0i; 0 ,\xi\right)   \frac{d\xi}{1-\xi}.
   \end{multline*}

Let $r_1(\theta)$ and $\iota_1(\theta)$ denote the real and imaginary parts of the quantity
$$
\int_{0}^{U_q^-(s(\theta)+0i)} \mathfrak{R}\left(s(\theta)+ 0i; 0 ,\xi\right)   \frac{d\xi}{1-\xi}, 
$$
respectively.  Similarly, let $r_2(\theta)$ and $\iota_2(\theta)$ be the real and imaginary parts of the quantity
$$
 \int_{\rho+q}^{U_q^-(s(\theta)+0i)}   \frac{y- \rho-1-\sigma_q^+}{q\sigma_q^+ +q+\rho} \mathfrak{R}(s(\theta)+ 0i; \rho+q,y)      dy , 
$$
respectively.  When $\theta$ is small
$$
\Delta I_3(s(\theta)) \sim \frac{2 iq}{\rho} (r_1(\theta)\iota_2(\theta)+\iota_1(\theta)r_2(\theta)).
$$

By definition and by using Lemma~\ref{lemmatech11} for $\zeta=0$
\begin{align*}
 2 i \iota_1(\theta) &=
\int_{U_q^+(s(\theta)+0i)}^{U_q^-(s(\theta)+0i)}  \mathfrak{R}\left(s(\theta)+ 0i; 0 ,\xi\right)   \frac{d\xi}{1-\xi}
= - \frac{ m(0,\theta) \pi i  e^{(( \Phi(0
,\theta) -\frac{\pi}{2})+\Phi(1,\theta) )  \cot\theta}   }{m(1,\theta) \cosh\left(\frac{\pi \cot\theta}{2}\right)}\\
&\sim -2i \pi \frac{q+\sqrt{\rho(1-q)}}{\sqrt{1-q}(\sqrt{1-q}-\sqrt{\rho})} e^{\frac{\sqrt{\rho}}{(\sqrt{1-q}-\sqrt{\rho})(q+\sqrt{\rho(1-q)}}}e^{-\frac{\pi}{\theta}}
\end{align*}
for small $\theta$. It follows that 
\begin{equation}
\iota_1(\theta) \sim -\frac{\pi (q+\sqrt{\rho(1-q)})e^{\frac{\sqrt{\rho(1-q)}}{q+\sqrt{\rho(1-q)}}+\frac{\sqrt{\rho}}{\sqrt{1-q}-\sqrt{\rho}}}}{1-q-\sqrt{\rho(1-q)}}e^{-\frac{\pi}{\theta}}, \label{iota1}
\end{equation}
when $\theta$ is small.

For the real part, we note that for $\xi \in [U_q^+(s(\theta)+0i),U_q^-(s(\theta)+0i)]$
$$
\left(1-\frac{\xi}{U_q^-(s(\theta)+0i)}\right)^{C_q^-(s(\theta)+0i)-1} =  \left(\frac{m(\xi,\theta)}{m(0,\theta)}e^{-i(\Phi(\xi,\theta)-\Phi(0,\theta)}\right)^{-\frac{1}{2}-\frac{i}{2}\cot\theta}
$$
It follows that when $\theta$ tend to 0,
\begin{multline*}
\int_{0}^{U_q^-(s(\theta)+0i)}  \mathfrak{R}\left(s(\theta)+ 0i; 0 ,\xi\right)   \frac{d\xi}{1-\xi} \sim \\ \int_{0}^{q+\sqrt{\rho(1-q)}} \frac{q+\sqrt{\rho(1-q)}}{q+\sqrt{\rho(1-q)}-\xi}e^{-  \frac{\sqrt{\rho(1-q)}\xi}{(q+\sqrt{\rho(1-q)})( q+\sqrt{\rho(1-q)}-\xi)}    }       \frac{d\xi}{1-\xi}.
\end{multline*}
Since the latter integral is real, a simple change of variable yields 
\begin{multline}
\lim_{\theta \to 0} r_1(\theta) = \mathcal{R}_1(q) \stackrel{def}{=}   \frac{q+\sqrt{\rho(1-q)}}{1-q- \sqrt{\rho(1-q)}}      e^{\frac{\sqrt{\rho(1-q)}}{ ( q+\sqrt{\rho(1-q)})(1-q- \sqrt{\rho(1-q)})}}  \\   \Gamma\left(0, \frac{\sqrt{\rho(1-q)}}{ ( q+\sqrt{\rho(1-q)})(1-q- \sqrt{\rho(1-q)})}  \right)     \label{r1}
\end{multline}

We now analyze the functions $r_2(\theta)$ and $\iota_2(\theta)$ for small $\theta$. By using the same arguments as above, we have as $\theta$ tends to 0
\begin{align*}
\mathfrak{R}\left(s(\theta)+0i;\rho+q,y\right) & \to  \frac{m(\rho+q,0)}{m(y,0)} e^{-\left( \frac{\sqrt{\rho(1-q}}{m(y,0)} - \frac{\sqrt{\rho(1-q}}{m(\rho+q,0)} \right)}\\
& = \frac{\sqrt{\rho(1-q)}-\rho}{q+\sqrt{\rho(1-q)}-y}e^{\frac{\sqrt{1-q}(\rho+q-y)}{(\sqrt{1-q}-\sqrt{\rho})(q+\sqrt{\rho(1-q)}-y)}}
\end{align*}
so that 
\begin{multline*}
r_2(\theta) \to \\  \frac{\sqrt{\rho(1-q)}-\rho}{(q +\sqrt{\rho(1-q)})^2}   \int_{\rho+q}^{q+\sqrt{\rho(1-q)}}   \frac{y- \rho-1-\sigma_q^+}{q+\sqrt{\rho(1-q)}-y}e^{-\frac{\sqrt{\rho(1-q)}(y-\rho-q)}{(\sqrt{\rho(1-q)}-{\rho})(q+\sqrt{\rho(1-q)}-y)}} dy.
\end{multline*}

By setting
$$
u =\frac{\sqrt{(1-q)}(y-\rho-q)}{(\sqrt{1-q}-\sqrt{\rho})(q+\sqrt{\rho(1-q)}-y)} 
$$
so that
$$
y(u) = \frac{(\sqrt{1-q}-\sqrt{\rho})(q+\sqrt{\rho(1-q)})u + \sqrt{1-q}(\rho+q)}{(\sqrt{1-q}-\sqrt{\rho})u + \sqrt{1-q}},
$$
we have
\begin{multline*}
 \int_{\rho+q}^{q+\sqrt{\rho(1-q)}}   \frac{y- \rho-1-\sigma_q^+}{q+\sqrt{\rho(1-q)}-y}e^{-\frac{\sqrt{\rho(1-q)}(y-\rho-q)}{(\sqrt{\rho(1-q)}-{\rho})(q+\sqrt{\rho(1-q)}-y)}} dy = \\
- \int_0^\infty e^{-u} dy(u) -\sqrt{\rho(1-q)} \int_0^\infty \frac{e^{-u}}{q+\sqrt{\rho(1-q)}-y(u)} dy(u). 
\end{multline*}
It is easily checked that
$$
d y(u) = \frac{\sqrt{\rho(1-q)}(\sqrt{1-q}-\sqrt{\rho})^2}{((\sqrt{1-q}-\sqrt{\rho})u + \sqrt{1-q})^2}du
$$
and
$$
\frac{dy(u)}{q+\sqrt{\rho(1-q)}-y(u)} = \frac{\sqrt{1-q}-\sqrt{\rho}}{(\sqrt{1-q}-\sqrt{\rho})u + \sqrt{1-q}}du.
$$
Now, an integration by parts yields
$$
 \int_0^\infty e^{-u} dy(u) = \sqrt{\rho(1-q)}-\rho-  \int_0^\infty \frac{\sqrt{\rho(1-q)}(\sqrt{1-q}-\sqrt{\rho})}{(\sqrt{1-q}-\sqrt{\rho})u + \sqrt{1-q}}    e^{-u} du
$$
and we eventually obtain
\begin{multline*}
 \int_{\rho+q}^{q+\sqrt{\rho(1-q)}}   \frac{y- \rho-1-\sigma_q^+}{q+\sqrt{\rho(1-q)}-y}e^{-\frac{\sqrt{\rho(1-q)}(y-\rho-q)}{(\sqrt{\rho(1-q)}-{\rho})(q+\sqrt{\rho(1-q)}-y)}} dy = \\
- (\sqrt{\rho(1-q)}-\rho)
\end{multline*}
so that
$$
r_2(\theta) \to -  \frac{(\sqrt{\rho(1-q)}-\rho)^2}{(q +\sqrt{\rho(1-q)})^2} 
$$
when $\theta$ tends to 0.

The imaginary part $\iota_2(\theta)$ is such that
$$
2 i \iota_2(\theta) \sim \int_{U_q^+(s(\theta)+0i)}^{U_q^-(s(\theta)+0i)} \frac{y- \rho-1-\sigma_q^+}{q\sigma_q^+ +q+\rho}     \mathfrak{R}\left(s(\theta)+0i;\rho+q,y\right) dy  
$$
for small $\theta$.
By using Equation~\eqref{repR} for $\zeta= \rho+q$, we have for small $\theta$
\begin{multline*}
2 i \iota_2(\theta) \sim \frac{i \rho  \sqrt{1-q}(\sqrt{1-q}-\sqrt{\rho}) e^{\frac{\sqrt{\rho(1-q)}}{\sqrt{\rho(1-q)}-\rho}}} {(q+\sqrt{\rho(1-q)})^2} e^{-\frac{\pi}{2\theta}}   \\  \int_0^1(1- 2i t \theta ) t^{-\frac{1}{2}-\frac{i \cot\theta}{2}}(1-t)^{-\frac{1}{2} +\frac{i\cot\theta}{2}} dt.
\end{multline*}

We have for small $\theta$
\begin{multline*}
 \int_0^1(1-2 it \theta ) t^{-\frac{1}{2}-\frac{i \cot\theta}{2}}(1-t)^{-\frac{1}{2} +\frac{i\cot\theta}{2}} dt = \\ \Gamma(\beta(\theta)\Gamma(1-\beta(\theta))F(-1,\beta(\theta),1;2i \theta) \sim -2i \pi  \theta  e^{-\frac{\pi}{2\theta}},
 \end{multline*}
 where $\beta(\theta)$ is defined by Equation~\eqref{betatheta}.  Hence $\iota_2(\theta) = o\left(e^{-\frac{\pi}{\theta}} \right)$ for small $\theta$.

It follows that for small $\theta$
$$
\Delta I_3(s(\theta)) \sim \frac{2 i \pi q (\sqrt{1-q}-\sqrt{\rho}) e^{\frac{\sqrt{\rho}}{(q+\sqrt{\rho(1-q)})(\sqrt{1-q}-\sqrt{\rho})}}}{\sqrt{1-q}(q+\sqrt{\rho(1-q)})}e^{-\frac{\pi}{\theta}}. 
$$
We deduce that for large $x$
\begin{multline*}
I_3(x) \sim \\  \frac{\rho+q}{2 i\pi \sigma_q^+} \int_0^\pi  \frac{2 i \pi q (\sqrt{1-q}-\sqrt{\rho}) e^{\frac{\sqrt{\rho}}{(q+\sqrt{\rho(1-q)})(\sqrt{1-q}-\sqrt{\rho})}}}{\sqrt{1-q}(q+\sqrt{\rho(1-q)})}e^{-\frac{\pi}{\theta}} 2 \sqrt{\rho(1-q)}\sin\theta e^{s(\theta)x} d\theta
\end{multline*}
and Equation~\eqref{asympI3} follows.

\bibliographystyle{plain}
\bibliography{biblio}
\end{document}